\documentclass[11pt,onecolumn]{article}%
\usepackage{amsmath}
\usepackage{amsfonts}
\usepackage{amssymb}
\usepackage{graphicx}
\usepackage{fullpage}%
\providecommand{\U}[1]{\protect\rule{.1in}{.1in}}
\newtheorem{theorem}{Theorem}

\newtheorem{conjecture}[theorem]{Conjecture}
\newtheorem{corollary}[theorem]{Corollary}

\newtheorem{definition}[theorem]{Definition}

\newtheorem{lemma}[theorem]{Lemma}

\newtheorem{problem}[theorem]{Problem}
\newtheorem{proposition}[theorem]{Proposition}

\newenvironment{proof}[1][Proof]{\noindent\textbf{#1.} }{\ \rule{0.5em}{0.5em}}
\begin{document}

\title{The Equivalence of Sampling and Searching}
\author{Scott Aaronson\thanks{MIT. \ Email: aaronson@csail.mit.edu. \ \ This material
is based upon work supported by the National Science Foundation under Grant
No. 0844626. \ Also supported by a DARPA YFA grant and the Keck Foundation.}}
\date{}
\maketitle

\begin{abstract}
In a \textit{sampling problem}, we are given an input $x\in\left\{
0,1\right\}  ^{n}$, and asked to sample approximately from a probability
distribution $\mathcal{D}_{x}$ over $\operatorname*{poly}\left(  n\right)
$-bit strings. \ In a \textit{search problem}, we are given an input
$x\in\left\{  0,1\right\}  ^{n}$, and asked to find a member of a nonempty set
$A_{x}$ with high probability. \ (An example is finding a Nash equilibrium.)
\ In this paper, we use tools from Kolmogorov complexity and algorithmic
information theory to show that sampling and search problems are essentially
equivalent. \ More precisely, for any sampling problem $S$, there exists a
search problem $R_{S}$\ such that, if $\mathcal{C}$ is any \textquotedblleft
reasonable\textquotedblright\ complexity class, then $R_{S}$\ is in the search
version of $\mathcal{C}$\ if and only if $S$\ is in the sampling version.

As one application, we show that $\mathsf{SampP}=\mathsf{SampBQP}$ if and only
if $\mathsf{FBPP}=\mathsf{FBQP}$: in other words, classical computers can
efficiently sample the output distribution of every quantum circuit, if and
only if they can efficiently solve every search problem that quantum computers
can solve. \ A second application is that, assuming a plausible conjecture,
there exists a search problem $R$\ that can be solved using a simple
linear-optics experiment, but that cannot be solved efficiently by a classical
computer unless the polynomial hierarchy collapses. \ That application will be
described in a forthcoming paper with Alex Arkhipov on the computational
complexity of linear optics.

\end{abstract}

\section{Introduction\label{INTRO}}

The \textit{Extended Church-Turing Thesis (ECT)} says that all computational
problems that are feasibly solvable in the physical world are feasibly
solvable by a probabilistic Turing machine. \ By now, there have been almost
two decades of discussion about this thesis, and the challenge that quantum
computing poses to it. \ This paper is about a related question that has
attracted surprisingly little interest: namely, what exactly should we
understand the ECT to \textit{state}? \ When we say \textquotedblleft all
computational problems,\textquotedblright\ do we mean decision problems?
promise problems? search problems? sampling problems? possibly other types of
problems?\ \ Could the ECT hold for some of these types of problems but fail
for others?

Our main result is an \textit{equivalence} between sampling and search
problems: the ECT holds for one type of problem if and only if it holds for
the other. \ As a motivating example, we will prove the surprising fact that,
if classical computers can efficiently solve any search problem that quantum
computers can solve, then they can \textit{also} approximately sample the
output distribution of any quantum circuit. \ The proof makes essential use of
Kolmogorov complexity. \ The technical tools that\ we will use are standard
ones in the algorithmic information theory literature; our contribution is
simply to apply those tools to obtain a useful equivalence principle in
complexity theory that seems not to have been known before.

While the \textit{motivation} for our equivalence theorem came from quantum
computing, we wish to stress that the theorem itself is much more general, and
has nothing to do with quantum computing in particular. \ Throughout this
paper, we will use the \textit{names} of quantum complexity classes---such as
$\mathsf{BQP}$\ (Bounded-Error Quantum Polynomial-Time), the class of
languages efficiently decidable by a quantum algorithm---but only as
\textquotedblleft black boxes.\textquotedblright\ \ No familiarity whatsoever
with quantum computing is needed.

The rest of the paper is organized as follows. \ Section \ref{BACKGROUND}
contains a general discussion of the relationships among decision problems,
promise problems, search problems, and sampling problems; it can be safely
skipped by readers already familiar with this material. \ Section
\ref{RESULTS} states our main result, as well as its implications for quantum
computing in general and linear-optics experiments in particular. \ Section
\ref{TECHNIQUES}\ explains how Kolmogorov complexity is used to prove the main
result, and situates the result in the context of earlier work on Kolmogorov
complexity. \ Next, in Section \ref{PRELIM}, we review some needed definitions
and results from complexity theory (in Section \ref{COMPLEXITY}), algorithmic
information theory (in Section \ref{KOLMOG}), and \textquotedblleft
standard\textquotedblright\ information theory (in Section \ref{INFOTHY}).
\ We prove the main result in Section \ref{MAIN}, and the example application
to quantum computing in Section \ref{IMPQC}. \ Finally, in Section \ref{OPEN},
we present several extensions and generalizations of the main result, which
address various shortcomings of it. \ Section \ref{OPEN} also discusses
numerous open problems.

\subsection{Background\label{BACKGROUND}}

Theoretical computer science has traditionally focused on \textit{language
decision problems}, where given a language $L\subseteq\left\{  0,1\right\}
^{\ast}$, the goal is to decide whether $x\in L$\ for any input $x$. \ From
this perspective, asking whether quantum computing contradicts the ECT is
tantamount to asking:

\begin{problem}
\label{bqpprob}Does $\mathsf{BPP}=\mathsf{BQP}$?
\end{problem}

However, one can also consider \textit{promise problems},\ where the goal is
to accept all inputs in a set $L_{\operatorname*{YES}}\subseteq\left\{
0,1\right\}  ^{\ast}$\ and reject all inputs in another set
$L_{\operatorname*{NO}}\subseteq\left\{  0,1\right\}  ^{\ast}$. \ Here
$L_{\operatorname*{YES}}$\ and $L_{\operatorname*{NO}}$\ are disjoint, but
their union is not necessarily all strings, and an algorithm can do whatever
it likes on inputs not in $L_{\operatorname*{YES}}\cup L_{\operatorname*{NO}}%
$. \ Goldreich \cite{goldreich:promise} has made a strong case that promise
problems are at least as fundamental as language decision problems, if not
more so. \ To give one relevant example, the task

\begin{quotation}
\textit{Given a quantum circuit }$C$\textit{, estimate the probability
}$p\left(  C\right)  $\textit{\ that }$C$\textit{\ accepts}
\end{quotation}

\noindent is easy to formulate as a promise problem, but has no known
formulation as a language decision problem. \ The reason is the usual
\textquotedblleft knife-edge\textquotedblright\ issue: given any probability
$p^{\ast}\in\left[  0,1\right]  $\ and error bound $\varepsilon\geq
1/\operatorname*{poly}\left(  n\right)  $, we can ask a simulation algorithm
to accept all quantum circuits $C$ such that $p\left(  C\right)  \geq p^{\ast
}+\varepsilon$, and to reject all circuits $C$\ such that $p\left(  C\right)
\leq p^{\ast}-\varepsilon$. \ But we cannot reasonably ask an algorithm to
decide whether $p\left(  C\right)  =p^{\ast}+2^{-n}$\ or $p\left(  C\right)
=p^{\ast}-2^{-n}$: if $p\left(  C\right)  $ is too close to $p^{\ast}$, then
the algorithm's behavior is unknown.

Let $\mathsf{P{}romiseBPP}$\ and $\mathsf{P{}romiseBQP}$\ be the classes of
promise problems solvable by probabilistic and quantum computers respectively,
in polynomial time and with bounded probability of error. \ Then a second way
to ask whether quantum mechanics contradicts the ECT is to ask:

\begin{problem}
\label{promisebqpprob}Does $\mathsf{P{}romiseBPP}=\mathsf{P{}romiseBQP}$?
\end{problem}

Now, if one accepts replacing languages by promise problems, then there seems
little reason not to go further. \ One can also consider \textit{search
problems}, where given an input $x\in\left\{  0,1\right\}  ^{n}$, the goal is
to output any element of some nonempty \textquotedblleft solution
set\textquotedblright\ $A_{x}\subseteq\left\{  0,1\right\}
^{\operatorname*{poly}\left(  n\right)  }$. \ (Search problems are also called
\textquotedblleft relational problems,\textquotedblright\ for the historical
reason that one can define such a problem using a binary relation
$R\subseteq\left\{  0,1\right\}  ^{\ast}\times\left\{  0,1\right\}  ^{\ast}$,
with $\left(  x,y\right)  \in R$\ if and only if $y\in A_{x}$. \ Another name
often used is \textquotedblleft function problems.\textquotedblright\ \ But
that is inaccurate,\ since the desired output is \textit{not} a function of
the input, except in the special case $\left\vert A_{x}\right\vert =1$. \ We
find \textquotedblleft search problems\textquotedblright\ to be the clearest
name, and will use it throughout in the hope that it sticks. \ The one
important point to remember is that a search problem need \textit{not} be an
$\mathsf{NP}$\ search problem: that is, solutions need not be efficiently verifiable.)

Perhaps the most famous example of a search problem is \textit{finding a Nash
equilibrium}, which Daskalakis et al. \cite{dgp}\ showed to be complete for
the class $\mathsf{PPAD}$. \ By Nash's Theorem, every game has at least one
Nash equilibrium, but the problem of finding one has no known formulation as
either a language decision problem or a promise problem.

Let $\mathsf{FBPP}$\ and $\mathsf{FBQP}$\ be the classes of search problems
solvable by probabilistic and quantum computers respectively, with success
probability $1-\delta$, in time polynomial in $n$ and $1/\delta$.\footnote{The
$\mathsf{F}$\ in $\mathsf{FBPP}$\ and $\mathsf{FBQP}$\ stands for
\textquotedblleft function problem.\textquotedblright\ \ Here we are following
the standard naming convention, even though the term \textquotedblleft
function problem\textquotedblright\ is misleading for the reason pointed out
earlier.} \ Then a third version of the \textquotedblleft ECT
question\textquotedblright\ is:

\begin{problem}
\label{fbqpprob}Does $\mathsf{FBPP}=\mathsf{FBQP}$?
\end{problem}

There is yet another important type of problem in theoretical computer
science. \ These are \textit{sampling problems}, where given an input
$x\in\left\{  0,1\right\}  ^{n}$, the goal is to sample (exactly or, more
often, approximately) from some probability distribution $\mathcal{D}_{x}%
$\ over $\operatorname*{poly}\left(  n\right)  $-bit strings. \ Well-known
examples of sampling problems include sampling a random point in a
high-dimensional convex body and sampling a random matching in a bipartite graph.

Let $\mathsf{SampP}$\ and $\mathsf{SampBQP}$\ be the classes of sampling
problems that are solvable by probabilistic and quantum computers
respectively, to within $\varepsilon$\ error in total variation distance, in
time polynomial in $n$\ and $1/\varepsilon$.\footnote{Note that we write
$\mathsf{SampP}$\ instead of \textquotedblleft$\mathsf{SampBPP}$%
\textquotedblright\ because there is no chance of confusion here. \ Unlike
with decision, promise, and relation problems, with sampling problems it does
not even make sense to talk about deterministic algorithms.} \ Then a fourth
version of our question is:

\begin{problem}
\label{sampbqpprob}Does $\mathsf{SampP}=\mathsf{SampBQP}$?
\end{problem}

Not surprisingly, \textit{all} of the above questions are open. But we can ask
an obvious meta-question:

\begin{quotation}
\noindent\textit{What is the relationship among Problems \ref{bqpprob}%
-\ref{sampbqpprob}? \ If the ECT fails in one sense, must it fail in the other
senses as well?}
\end{quotation}

\noindent In one direction, there are some easy implications:%
\begin{align*}
\mathsf{SampP}=\mathsf{SampBQP}  &  \Longrightarrow\mathsf{FBPP}%
=\mathsf{FBQP}\\
&  \Longrightarrow\mathsf{P{}romiseBPP}=\mathsf{P{}romiseBQP}\\
&  \Longrightarrow\mathsf{BPP}=\mathsf{BQP}.
\end{align*}
For the first implication, if every quantumly samplable distribution were also
classically samplable, then given a quantum algorithm $Q$\ solving a search
problem $R$, we could approximate $Q$'s\ output distribution using a classical
computer, and thereby solve $R$\ classically as well. \ For the second and
third implications, every promise problem is also a search problem (with
solution set $A_{x}\subseteq\left\{  0,1\right\}  $), and every language
decision problem is also a promise problem (with the empty promise).

So the interesting part concerns the possible implications in the
\textquotedblleft other\textquotedblright\ direction. \ For example, could it
be the case that $\mathsf{BPP}=\mathsf{BQP}$, yet $\mathsf{P{}romiseBPP}%
\neq\mathsf{P{}romiseBQP}$? \ Not only is this a formal possibility, but it
does not even seem absurd, when we consider that

\begin{enumerate}
\item[(1)] the existing candidates for languages in $\mathsf{BQP}%
\setminus\mathsf{BPP}$ (for example, decision versions of the factoring and
discrete log problems \cite{shor}) are all extremely \textquotedblleft
special\textquotedblright\ in nature, but

\item[(2)] $\mathsf{P{}romiseBQP}$ contains the \textquotedblleft
general\textquotedblright\ problem of estimating the acceptance probability of
an arbitrary quantum circuit.
\end{enumerate}

A second example of a difficult and unsolved meta-question is whether
$\mathsf{P{}romiseBPP}=\mathsf{P{}romiseBQP}$ implies $\mathsf{SampP}%
=\mathsf{SampBQP}$. \ Translated into \textquotedblleft physics
language,\textquotedblright\ the question is this: suppose we had an efficient
classical algorithm to estimate the \textit{expectation value} of any
observable\ in quantum mechanics. \ Would that imply an efficient classical
algorithm to \textit{simulate any quantum experiment},\ in the sense of
sampling from a probability distribution close to the one quantum mechanics
predicts? \ The difficulty is that, if we consider a quantum system of $n$
particles, then a measurement could in general have $c^{n}$\ possible
outcomes, each with probability on the order of $c^{-n}$. \ So, even supposing
we could estimate any \textit{given} probability to within $\pm\varepsilon$,
in time polynomial in $n$\ and $1/\varepsilon$, that would seem to be of
little help for the sampling task.

\subsection{Our Results\label{RESULTS}}

This paper shows that \textit{two} of the four types of problem discussed
above---namely, sampling problems and search problems---are essentially
equivalent. \ More precisely, given any sampling problem $S$, we will
construct a search problem $R=R_{S}$\ such that, if $\mathcal{C}$ is any
\textquotedblleft reasonable\textquotedblright\ model of computation, then $S$
is in $\mathsf{Samp}\mathcal{C}$\ (the sampling version of $\mathcal{C}$) if
and only if $R$ is in $\mathsf{F}\mathcal{C}$\ (the search version of
$\mathcal{C}$). \ Here is a more formal statement of the result:

\begin{theorem}
[Sampling/Searching Equivalence Theorem]\label{samprel}Let $S$\ be any
sampling problem. \ Then there exists a search problem $R_{S}$\ such that

\begin{enumerate}
\item[(i)] If $\mathcal{O}$\ is any oracle for $S$, then $R_{S}\in
\mathsf{FBPP}^{\mathcal{O}}$.

\item[(ii)] If $B$ is any probabilistic Turing machine solving $R_{S}$, then
$S\in\mathsf{SampP}^{B}$.
\end{enumerate}
\end{theorem}

As one application, we show that the \textquotedblleft
obvious\textquotedblright\ implication $\mathsf{SampP}=\mathsf{SampBQP}%
\Longrightarrow\mathsf{FBPP}=\mathsf{FBQP}$ can be reversed:

\begin{theorem}
\label{fbqpthm}$\mathsf{FBPP}=\mathsf{FBQP}$ if and only if $\mathsf{SampP}%
=\mathsf{SampBQP}$. \ In other words, classical computers can efficiently
solve every $\mathsf{FBQP}$ search problem, if and only if they can
approximately sample the output distribution of every quantum circuit.
\end{theorem}

As a second application (which was actually the original motivation for this
work), we are able to extend a recent result of Aaronson and Arkhipov
\cite{aark}. \ These authors give a sampling problem that is solvable using a
simple linear-optics experiment (so in particular, in $\mathsf{SampBQP}$), but
is \textit{not} solvable efficiently by a classical computer, unless the
permanent of a Gaussian random matrix can be approximated in $\mathsf{BPP}%
^{\mathsf{NP}}$. \ More formally, consider the following problem, called
$\left\vert \text{\textsc{GPE}}\right\vert ^{2}$\ (the \textsc{GPE}\ stands
for Gaussian Permanent Estimation):

\begin{problem}
[$\left\vert \text{\textsc{GPE}}\right\vert ^{2}$]\label{gpe2}Given an input
of the form $\left\langle X,0^{1/\varepsilon},0^{1/\delta}\right\rangle $,
where $X\in\mathbb{C}^{n\times n}$ is an $n\times n$ matrix of independent
$\mathcal{N}\left(  0,1\right)  $\ Gaussians, output a real number $y$\ such
that%
\[
\left\vert y-\left\vert \operatorname*{Per}\left(  X\right)  \right\vert
^{2}\right\vert \leq\varepsilon\cdot n!,
\]
with probability at least $1-\delta$ over both $X\sim\mathcal{N}\left(
0,1\right)  _{\mathbb{C}}^{n\times n}$\ and any internal randomness used by
the algorithm.
\end{problem}

Here $0^{1/\varepsilon}$\ and $0^{1/\delta}$\ represent the numbers
$1/\varepsilon$\ and $1/\delta$\ respectively encoded in unary; such unary
encoding is a standard trick for forcing an algorithm's running time to be
polynomial in $1/\varepsilon$\ and $1/\delta$\ as well as $n$.

The main result of \cite{aark}\ is the following:

\begin{theorem}
[Aaronson and Arkhipov \cite{aark}]\label{aathm}$\mathsf{SampP}%
=\mathsf{SampBQP}$ implies $\left\vert \text{\textsc{GPE}}\right\vert ^{2}%
\in\mathsf{FBPP}^{\mathsf{NP}}$.
\end{theorem}

Note that Theorem \ref{aathm}\ relativizes: for all oracles $\mathcal{O}$, if
$\mathsf{SampBQP}\subseteq\mathsf{SampBPP}^{\mathcal{O}}$, then $\left\vert
\text{\textsc{GPE}}\right\vert ^{2}\in\mathsf{FBPP}^{\mathsf{NP}^{\mathcal{O}%
}}$.

The central conjecture made in \cite{aark} is that estimating $\left\vert
\operatorname*{Per}\left(  X\right)  \right\vert ^{2}$ is as hard for a
Gaussian random matrix $X$\ as it is for an arbitrary matrix $X\in
\mathbb{C}^{n\times n}$:

\begin{conjecture}
[\cite{aark}]$\left\vert \text{\textsc{GPE}}\right\vert ^{2}$\ is
$\mathsf{\#P}$-complete.\label{gpeconj}
\end{conjecture}

Much of \cite{aark} is devoted to giving evidence for Conjecture \ref{gpeconj}.

Notice that, if Conjecture \ref{gpeconj}\ holds, then combining it with
Theorem \ref{aathm}, we find that $\mathsf{SampP}=\mathsf{SampBQP}$\ implies
$\mathsf{P}^{\mathsf{\#P}}=\mathsf{BPP}^{\mathsf{NP}}$ (which in turn implies
$\mathsf{PH}=\mathsf{BPP}^{\mathsf{NP}}$\ by Toda's Theorem \cite{toda}). \ Or
to put it differently: assuming Conjecture \ref{gpeconj}, there can be no
polynomial-time classical algorithm to sample (even approximately) the output
distribution of quantum circuits in general, or the linear-optics experiment
of \cite{aark} in particular, unless the polynomial hierarchy collapses to
$\mathsf{BPP}^{\mathsf{NP}}$. \ This can be taken as a surprising new form of
evidence against the Extended Church-Turing Thesis---assuming, of course, that
one is willing to state the ECT in terms of sampling problems.

Now, by using Theorem \ref{fbqpthm} from this paper, we can deduce, in a
completely \textquotedblleft automatic\textquotedblright\ way, that the
counterpart of Theorem \ref{aathm} holds with \textit{search} problems in
place of sampling problems:

\begin{corollary}
\label{aacor}$\mathsf{FBPP}=\mathsf{FBQP}$\ implies $\left\vert
\text{\textsc{GPE}}\right\vert ^{2}\in\mathsf{FBPP}^{\mathsf{NP}}$. \ So in
particular, assuming $\left\vert \text{\textsc{GPE}}\right\vert ^{2}$\ is
$\mathsf{\#P}$-complete\ and $\mathsf{PH}$\ is infinite, it follows that
$\mathsf{FBPP}\neq\mathsf{FBQP}$.
\end{corollary}

Indeed, assuming $\left\vert \text{\textsc{GPE}}\right\vert ^{2}$\ is
$\mathsf{\#P}$-complete, we cannot even have $\mathsf{FBQP}\subseteq
\mathsf{FBPP}^{\mathsf{PH}}$, unless $\mathsf{P}^{\mathsf{\#P}}=\mathsf{PH}%
$\ and the polynomial hierarchy collapses. \ To strengthen Corollary
\ref{aacor}\ still further, notice that one can replace $\mathsf{FBQP}$\ by
\textquotedblleft the class of search problems efficiently solvable with the
help of a linear-optics computer,\textquotedblright\ which is almost certainly
a proper subclass of $\mathsf{FBQP}$.

\subsection{Proof Overview\label{TECHNIQUES}}

Let us explain the basic difficulty we need to overcome to prove Theorem
\ref{samprel}. \ Given a probability distribution $\mathcal{D}_{x}$ over
$\left\{  0,1\right\}  ^{\operatorname*{poly}\left(  n\right)  }$, we want to
define a set $A_{x}\subseteq\left\{  0,1\right\}  ^{\operatorname*{poly}%
\left(  n\right)  }$, such that the ability to \textit{find} an element of
$A_{x}$ is equivalent to the ability to \textit{sample} from $\mathcal{D}_{x}%
$. \ At first glance, such a general reduction seems impossible. \ For let
$R=\left\{  A_{x}\right\}  _{x}$ be the search problem in which the goal is to
find an element of $A_{x}$\ given $x$. \ Then consider an oracle $\mathcal{O}$
that, on input $x$, returns the lexicographically first element of $A_{x}$.
\ Such an oracle $\mathcal{O}$ certainly solves $R$, but it seems useless\ if
our goal is to \textit{sample }uniformly from the set $A_{x}$ (or indeed, from
any other interesting distribution related to $A_{x}$).

Our solution will require going outside the black-box reduction
paradigm.\footnote{This was previously done for different reasons in a
cryptographic context---see for example Barak's beautiful PhD thesis
\cite{barak:thesis}.} \ In other words, given a sampling problem $S=\left\{
\mathcal{D}_{x}\right\}  _{x}$, we do \textit{not} show that $S\in
\mathsf{SampP}^{\mathcal{O}}$, where $\mathcal{O}$\ is any oracle that solves
the corresponding search problem $R_{S}$. \ Instead, we use the fact that
$\mathcal{O}$ is computed by a Turing machine. \ We then define $R_{S}$\ in
such a way that $\mathcal{O}$\ must return, not just any element in the
support of $\mathcal{D}_{x}$, but an element with \textit{near-maximal
Kolmogorov complexity}.

The idea here is simple: if a Turing machine $B$ is probabilistic, then it can
certainly output a string $x$ with high Kolmogorov complexity, by just
generating $x$ at random. \ But the converse also holds: if $B$ outputs a
string $x$\ with high Kolmogorov complexity, then $x$ \textit{must} have been
generated randomly. \ For otherwise, the code of $B$ would constitute a
succinct description of $x$.

Given any set $A\subseteq\left\{  0,1\right\}  ^{n}$, it is not hard to use
the above \textquotedblleft Kolmogorov trick\textquotedblright\ to force a
probabilistic Turing machine $B$ to sample almost-uniformly from $A$. \ We
simply ask $B$\ to produce $k$\ samples $x_{1},\ldots,x_{k}\in A$, for some
$k=\operatorname*{poly}\left(  n\right)  $, such that the tuple $\left\langle
x_{1},\ldots,x_{k}\right\rangle $\ has Kolmogorov complexity close to
$k\log_{2}\left\vert A\right\vert $. \ Then we output $x_{i}$\ for a uniformly
random $i\in\left[  k\right]  $.

However, one can also generalize the idea, to force $B$\ to sample from an
\textit{arbitrary} distribution $\mathcal{D}$, not necessarily uniform. \ One
way of doing this would be to reduce to the uniform case, by dividing the
support of $\mathcal{D}$\ into $\operatorname*{poly}\left(  n\right)
$\ \textquotedblleft buckets,\textquotedblright\ such that $\mathcal{D}$\ is
nearly-uniform within each bucket, and then asking $B$ to output
Kolmogorov-random elements in each bucket. \ In this paper, however, we will
follow a more direct approach, which exploits the beautiful known connection
between Kolmogorov complexity and Shannon information. \ In particular, we
will use the notion of a \textit{universal randomness test} from algorithmic
information theory \cite{livitanyi,gacs}. \ Let $\mathcal{U}$\ be the
\textquotedblleft universal prior,\textquotedblright\ in which each string
$x\in\left\{  0,1\right\}  ^{\ast}$ occurs with probability proportional to
$2^{-K\left(  x\right)  }$, where $K\left(  x\right)  $\ is the prefix-free
Kolmogorov complexity of $x$. \ Then given any computable distribution
$\mathcal{D}$ and fixed string $x$, the universal randomness test\ provides a
way to decide whether $x$\ was \textquotedblleft plausibly drawn from
$\mathcal{D}$,\textquotedblright\ by considering the ratio $\Pr_{\mathcal{D}%
}\left[  x\right]  /\Pr_{\mathcal{U}}\left[  x\right]  $. \ The main technical
fact we need to prove is simply that such a test can be applied in our
\textit{complexity-theoretic} context, where we care (for example) that the
number of samples from $\mathcal{D}$\ scales polynomially with the inverses of
the relevant error parameters.

From one perspective, our result represents a surprising use of Kolmogorov
complexity in the seemingly \textquotedblleft distant\textquotedblright\ realm
of polynomial-time reductions. \ Let us stress that we are \textit{not} using
Kolmogorov complexity as just a technical convenience, or as shorthand for a
counting argument. \ Rather, Kolmogorov complexity seems essential even to
define a search problem $R_{S}$\ with the properties we need. \ From another
perspective, however, our use of Kolmogorov complexity is close in spirit to
the reasons why Kolmogorov complexity was defined and studied in the first
place! \ The whole point, after all, is to be able to talk about the
\textquotedblleft randomness of an individual object,\textquotedblright%
\ without reference to any distribution from which the object was drawn. \ And
that is exactly what we need, if we want to achieve the \textquotedblleft
paradoxical\textquotedblright\ goal of sampling from a distribution, using an
oracle that is guaranteed only to output a \textit{fixed} string $x$\ with
specified properties.

\section{Preliminaries\label{PRELIM}}

\subsection{Sampling and Search Problems\label{COMPLEXITY}}

We first formally define sampling problems, as well as the complexity classes
$\mathsf{SampP}$\ and $\mathsf{SampBQP}$ of sampling problems that are
efficiently solvable by classical and quantum computers respectively.

\begin{definition}
[Sampling Problems, $\mathsf{SampP}$, and $\mathsf{SampBQP}$]A
\textit{sampling problem} $S$ is a collection of probability distributions
$\left(  \mathcal{D}_{x}\right)  _{x\in\left\{  0,1\right\}  ^{\ast}}$, one
for each input string $x\in\left\{  0,1\right\}  ^{n}$, where $\mathcal{D}%
_{x}$\ is a distribution over $\left\{  0,1\right\}  ^{p\left(  n\right)  }$,
for some fixed polynomial $p$. \ Then $\mathsf{SampP}$ is the class of
sampling problems $S=\left(  \mathcal{D}_{x}\right)  _{x\in\left\{
0,1\right\}  ^{\ast}}$ for which there exists a probabilistic polynomial-time
algorithm $B$\ that, given $\left\langle x,0^{1/\varepsilon}\right\rangle $ as
input, samples from a probability distribution $\mathcal{C}_{x}$\ such that
$\left\Vert \mathcal{C}_{x}-\mathcal{D}_{x}\right\Vert \leq\varepsilon$.
\ $\mathsf{SampBQP}$\ is defined the same way, except that $B$\ is a quantum
algorithm rather than a classical one.
\end{definition}

Let us also define search problems, as well as the complexity classes
$\mathsf{FBPP}$\ and $\mathsf{FBQP}$ of search problems that are efficiently
solvable by classical and quantum computers respectively.

\begin{definition}
[Search Problems, $\mathsf{FBPP}$, and $\mathsf{FBQP}$]A search problem $R$ is
a collection of nonempty sets $\left(  A_{x}\right)  _{x\in\left\{
0,1\right\}  ^{\ast}}$, one for each input string $x\in\left\{  0,1\right\}
^{n}$, where $A_{x}\subseteq\left\{  0,1\right\}  ^{p\left(  n\right)  }$ for
some fixed polynomial $p$. \ Then$\ \mathsf{FBPP}$ is the class of search
problems $R=\left(  A_{x}\right)  _{x\in\left\{  0,1\right\}  ^{\ast}}$\ for
which there exists a probabilistic polynomial-time algorithm $B$ that, given
an input $x\in\left\{  0,1\right\}  ^{n}$ together with $0^{1/\varepsilon}$,
produces an output $y$ such that%
\[
\Pr\left[  y\in A_{x}\right]  \geq1-\varepsilon,
\]
where the probability is over $B$'s internal randomness. \ $\mathsf{FBQP}$\ is
defined the same way, except that $B$ is a quantum algorithm rather than a
classical one.
\end{definition}

\subsection{Algorithmic Information Theory\label{KOLMOG}}

We now review some basic definitions and results from the theory of Kolmogorov
complexity. \ Recall that a set of strings $P\subset\left\{  0,1\right\}
^{\ast}$\ is called \textit{prefix-free} if no $x\in P$\ is a prefix of any
other $y\in P$.

\begin{definition}
[Kolmogorov complexity]\label{kolmogdef}Fix a universal Turing machine $U$,
such that the set of valid programs of $U$ is prefix-free. \ Then $K\left(
y\right)  $, or the prefix-free Kolmogorov complexity of $y$, is the minimum
length of a program $x$\ such that $U\left(  x\right)  =y$. \ We can also
define the conditional Kolmogorov complexity\ $K\left(  y|z\right)  $, as the
minimum length of a program $x$ such that $U\left(  \left\langle
x,z\right\rangle \right)  =y$.
\end{definition}

We are going to need two basic lemmas that relate Kolmogorov complexity to
standard information theory, and that can be found in the book of Li and
Vit\'{a}nyi \cite{livitanyi} for example. \ The first lemma follows almost
immediately from Shannon's noiseless channel coding theorem.

\begin{lemma}
\label{kolmogsmall}Let $\mathcal{D}=\left\{  p_{x}\right\}  $\ be any
computable distribution over strings, and let $x$\ be any element in the
support of $\mathcal{D}$. \ Then%
\[
K\left(  x\right)  \leq\log_{2}\frac{1}{p_{x}}+K\left(  \mathcal{D}\right)
+O\left(  1\right)  ,
\]
where $K\left(  \mathcal{D}\right)  $\ represents the length of the shortest
program to sample from $\mathcal{D}$. \ The same holds if we replace $K\left(
x\right)  $\ and $K\left(  \mathcal{D}\right)  $\ by $K\left(  x|y\right)
$\ and $K\left(  \mathcal{D}|y\right)  $\ respectively, for any fixed $y$.
\end{lemma}

The next lemma follows from a counting argument.

\begin{lemma}
[\cite{livitanyi}]\label{kolmogbig}Let $\mathcal{D}=\left\{  p_{x}\right\}
$\ be any distribution over strings (not necessarily computable). \ Then there
exists a universal constant $b$ such that%
\[
\Pr_{x\sim\mathcal{D}}\left[  K\left(  x\right)  \geq\log_{2}\frac{1}{p_{x}%
}-c\right]  \geq1-\frac{b}{2^{c}}.
\]
The same holds if we replace $K\left(  x\right)  $\ by $K\left(  x|y\right)
$\ for any fixed $y$.
\end{lemma}

\subsection{Information Theory\label{INFOTHY}}

This section reviews some basic definitions and facts from information theory.
\ Let $\mathcal{A}=\left\{  p_{x}\right\}  _{x}$\ and $\mathcal{B}=\left\{
q_{x}\right\}  _{x}$\ be two probability distributions over $\left[  N\right]
$. \ Then recall that the \textit{variation distance} between $\mathcal{A}%
$\ and $\mathcal{B}$\ is defined as%
\[
\left\Vert \mathcal{A}-\mathcal{B}\right\Vert :=\frac{1}{2}\sum_{i=1}%
^{N}\left\vert p_{x}-q_{x}\right\vert ,
\]
while the \textit{KL-divergence} is%
\[
D_{KL}\left(  \mathcal{A}||\mathcal{B}\right)  :=\sum_{i=1}^{N}p_{x}\log
_{2}\frac{p_{x}}{q_{x}}.
\]
The variation distance and the KL-divergence are related as follows:

\begin{proposition}
[Pinsker's Inequality]\label{pinsker}$\left\Vert \mathcal{A}-\mathcal{B}%
\right\Vert \leq\sqrt{2D_{KL}\left(  \mathcal{A}||\mathcal{B}\right)  }.$
\end{proposition}

We will also need a fact about KL-divergence that has been useful in the study
of parallel repetition, and that can be found (for example) in a paper by Rao
\cite{rao}.

\begin{proposition}
[\cite{rao}]\label{anup}Let $\mathcal{R}$\ be a distribution over $\left[
N\right]  ^{k}$, with marginal distribution $\mathcal{R}_{i}$\ \ on the
$i^{th}$\ coordinate. \ Let $\mathcal{D}$\ be a distribution over $\left[
N\right]  $. \ Then%
\[
\sum_{i=1}^{k}D_{KL}\left(  \mathcal{R}_{i}||\mathcal{D}\right)  \leq
D_{KL}\left(  \mathcal{R}||\mathcal{D}^{k}\right)
\]

\end{proposition}

\section{Main Result\label{MAIN}}

Let $S=\left\{  \mathcal{D}_{x}\right\}  _{x}$ be a sampling problem. \ Then
our goal is to construct a search problem $R=R_{S}=\left\{  A_{x}\right\}
_{x}$ that is \textquotedblleft equivalent\textquotedblright\ to $S$. \ Given
an input of the form $\left\langle x,0^{1/\delta}\right\rangle $,\ the goal in
the search problem will be to produce an output $Y$\ such that $Y\in
A_{x,\delta}$, with success probability at least $1-\delta$. \ The running
time should be $\operatorname*{poly}\left(  n,1/\delta\right)  $.

Fix an input $x\in\left\{  0,1\right\}  ^{n}$, and let $\mathcal{D}%
:=\mathcal{D}_{x}$\ be the corresponding probability distribution over
$\left\{  0,1\right\}  ^{m}$. \ Let $p_{y}:=\Pr_{\mathcal{D}}\left[  y\right]
$\ be the probability of $y$. \ We now define the search problem $R$. \ Let
$N:=m/\delta^{2.1}$, and let $Y=\left\langle y_{1},\ldots,y_{N}\right\rangle
$\ be an $N$-tuple of $m$-bit strings. \ Then we set $Y\in A_{x,\delta}$\ if
and only if%
\[
\log_{2}\frac{1}{p_{y_{1}}\cdots p_{y_{N}}}\leq K\left(  Y~|~x,\delta\right)
+\beta,
\]
where $\beta:=1+\log_{2}1/\delta$.

The first thing we need to show is that any algorithm that solves the sampling
problem $S$ also solves the search problem $R$ with high probability.

\begin{lemma}
\label{stor}Let $\mathcal{C}=\mathcal{C}_{x}$ be any distribution over
$\left\{  0,1\right\}  ^{m}$ such that\ $\left\Vert \mathcal{C}-\mathcal{D}%
\right\Vert \leq\varepsilon$. \ Then%
\[
\Pr_{Y\sim\mathcal{C}^{N}}\left[  Y\notin A_{x,\delta}\right]  \leq\varepsilon
N+\frac{b}{2^{\beta}}.
\]

\end{lemma}

\begin{proof}
We have%
\begin{align*}
\Pr_{Y\sim\mathcal{C}^{N}}\left[  Y\notin A_{x,\delta}\right]   &  \leq
\Pr_{Y\sim\mathcal{D}^{N}}\left[  Y\notin A_{x,\delta}\right]  +\left\Vert
\mathcal{C}^{N}-\mathcal{D}^{N}\right\Vert \\
&  \leq\Pr_{Y\sim\mathcal{D}^{N}}\left[  Y\notin A_{x,\delta}\right]
+\varepsilon N.
\end{align*}
So it suffices to consider a $Y$\ drawn from $\mathcal{D}^{N}$. \ By Lemma
\ref{kolmogbig},%
\[
\Pr_{Y\sim\mathcal{D}^{N}}\left[  K\left(  Y~|~x,\delta\right)  \geq\log
_{2}\frac{1}{p_{y_{1}}\cdots p_{y_{N}}}-\beta\right]  \geq1-\frac{b}{2^{\beta
}}%
\]
Therefore%
\[
\Pr_{Y\sim\mathcal{D}^{N}}\left[  Y\notin A_{x,\delta}\right]  \leq\frac
{b}{2^{\beta}},
\]
and we are done.
\end{proof}

The second thing we need to show is that any algorithm that solves the search
problem $R$ also samples from a distribution that is close to $\mathcal{D}%
$\ in variation distance.

\begin{lemma}
\label{rtos}Let $B$ be a probabilistic Turing machine, which given input
$\left\langle x,0^{1/\delta}\right\rangle $\ outputs an $N$-tuple
$Y=\left\langle y_{1},\ldots,y_{N}\right\rangle $\ of $m$-bit strings.
\ Suppose that%
\[
\Pr\left[  B\left(  x,0^{1/\delta}\right)  \in A_{x,\delta}\right]
\geq1-\delta,
\]
where the probability is over $B$'s internal randomness. \ Let $\mathcal{R}%
=\mathcal{R}_{x}$\ be the distribution over outputs of $B\left(  x\right)  $,
and let $\mathcal{C}=\mathcal{C}_{x}$\ be the distribution over $\left\{
0,1\right\}  ^{m}$\ that is obtained by from $\mathcal{R}$ by choosing one of
the $y_{i}$'s uniformly at random. \ Then there exists a constant $Q_{B}$,
depending on $B$, such that%
\[
\left\Vert \mathcal{C}-\mathcal{D}\right\Vert \leq\delta+Q_{B}\sqrt
{\frac{\beta}{N}}.
\]

\end{lemma}

\begin{proof}
Let $\mathcal{R}^{\prime}$\ be a distribution that is identical to
$\mathcal{R}$, except that we condition on\ $B\left(  x,0^{1/\delta}\right)
\in A_{x,\delta}$. \ Then by hypothesis, $\left\Vert \mathcal{R}%
-\mathcal{R}^{\prime}\right\Vert \leq\delta$. \ Now let $\mathcal{R}%
_{i}^{\prime}$\ be the marginal distribution of $\mathcal{R}^{\prime}$\ on the
$i^{th}$\ coordinate, and let%
\[
\mathcal{C}^{\prime}=\frac{1}{N}\sum_{i=1}^{N}\mathcal{R}_{i}^{\prime}%
\]
be the distribution over $\left\{  0,1\right\}  ^{m}$\ that is obtained from
$\mathcal{R}^{\prime}$\ by choosing one of the $y_{i}$'s uniformly at random.
\ Then clearly $\left\Vert \mathcal{C}-\mathcal{C}^{\prime}\right\Vert
\leq\delta$ as well. \ So by the triangle inequality,%
\begin{align*}
\left\Vert \mathcal{C}-\mathcal{D}\right\Vert  &  \leq\left\Vert
\mathcal{C}-\mathcal{C}^{\prime}\right\Vert +\left\Vert \mathcal{C}^{\prime
}-\mathcal{D}\right\Vert \\
&  \leq\delta+\left\Vert \mathcal{C}^{\prime}-\mathcal{D}\right\Vert ,
\end{align*}
and it suffices to upper-bound $\left\Vert \mathcal{C}^{\prime}-\mathcal{D}%
\right\Vert $.

Let $q_{Y}:=\Pr_{\mathcal{R}^{\prime}}\left[  Y\right]  $. \ Then by Lemma
\ref{kolmogsmall},%
\[
K\left(  Y~|~x,\delta\right)  \leq\log_{2}\frac{1}{q_{Y}}+K\left(
\mathcal{R}^{\prime}\right)  +O\left(  1\right)
\]
for all $Y\in\left(  \left\{  0,1\right\}  ^{m}\right)  ^{N}$. \ Also, since
$Y\in A_{x,\delta}$, by assumption we have%
\[
\log_{2}\frac{1}{p_{y_{1}}\cdots p_{y_{N}}}\leq K\left(  Y~|~x,\delta\right)
+\beta.
\]
Combining,%
\[
\log_{2}\frac{1}{p_{y_{1}}\cdots p_{y_{N}}}\leq\log_{2}\frac{1}{q_{Y}%
}+K\left(  \mathcal{R}^{\prime}\right)  +O\left(  1\right)  +\beta.
\]
This implies the following upper bound on the KL-divergence:%
\begin{align*}
D_{KL}\left(  \mathcal{R}^{\prime}||\mathcal{D}^{N}\right)   &  =\sum
_{Y\in\left(  \left\{  0,1\right\}  ^{m}\right)  ^{N}}q_{Y}\log_{2}\frac
{q_{Y}}{p_{y_{1}}\cdots p_{y_{N}}}\\
&  \leq\max_{Y}\log_{2}\frac{q_{Y}}{p_{y_{1}}\cdots p_{y_{N}}}\\
&  \leq K\left(  \mathcal{R}^{\prime}\right)  +O\left(  1\right)  +\beta.
\end{align*}
So by Proposition \ref{anup},%
\[
\sum_{i=1}^{N}D_{KL}\left(  \mathcal{R}_{i}^{\prime}||\mathcal{D}\right)  \leq
D_{KL}\left(  \mathcal{R}^{\prime}||\mathcal{D}^{N}\right)  \leq K\left(
\mathcal{R}^{\prime}\right)  +O\left(  1\right)  +\beta,
\]
and by Proposition \ref{pinsker},%
\[
\frac{1}{2}\sum_{i=1}^{N}\left\Vert \mathcal{R}_{i}^{\prime}-\mathcal{D}%
\right\Vert ^{2}\leq K\left(  \mathcal{R}^{\prime}\right)  +O\left(  1\right)
+\beta.
\]
So by Cauchy-Schwarz,%
\[
\sum_{i=1}^{N}\left\Vert \mathcal{R}_{i}^{\prime}-\mathcal{D}\right\Vert
\leq\sqrt{N\left(  2\beta+2K\left(  \mathcal{R}^{\prime}\right)  +O\left(
1\right)  \right)  }.
\]
Hence%
\[
\left\Vert \mathcal{C}^{\prime}-\mathcal{D}\right\Vert \leq\sqrt{\frac
{2\beta+2K\left(  \mathcal{R}^{\prime}\right)  +O\left(  1\right)  }{N}},
\]
and%
\begin{align*}
\left\Vert \mathcal{C}-\mathcal{D}\right\Vert  &  \leq\left\Vert
\mathcal{C}-\mathcal{C}^{\prime}\right\Vert +\left\Vert \mathcal{C}^{\prime
}-\mathcal{D}\right\Vert \\
&  \leq\delta+\sqrt{\frac{2\beta+2K\left(  \mathcal{R}^{\prime}\right)
+O\left(  1\right)  }{N}}\\
&  \leq\delta+Q_{B}\sqrt{\frac{\beta}{N}},
\end{align*}
for some constant $Q_{B}$\ depending on $B$.
\end{proof}

By combining Lemmas \ref{stor}\ and \ref{rtos}, we can now prove Theorem
\ref{samprel}: that for any sampling problem $S=\left(  \mathcal{D}%
_{x}\right)  _{x\in\left\{  0,1\right\}  ^{\ast}}$\ (where $\mathcal{D}_{x}%
$\ is a distribution over $m=m\left(  n\right)  $-bit strings), there exists a
search problem $R_{S}=\left(  A_{x}\right)  _{x\in\left\{  0,1\right\}
^{\ast}}$\ that is \textquotedblleft equivalent\textquotedblright\ to $S$ in
the following two senses.

\begin{enumerate}
\item[(i)] Let $\mathcal{O}$\ be any oracle that, given $\left\langle
x,0^{1/\varepsilon},r\right\rangle $\ as input, outputs a sample from a
distribution $\mathcal{C}_{x}$\ such that $\left\Vert \mathcal{C}%
_{x}-\mathcal{D}_{x}\right\Vert \leq\varepsilon$, as we vary the random string
$r$. \ Then$\ R_{S}\in\mathsf{FBPP}^{\mathcal{O}}$.

\item[(ii)] Let $B$\ be any probabilistic Turing machine that, given
$\left\langle x,0^{1/\delta}\right\rangle $\ as input, outputs a $Y\in\left(
\left\{  0,1\right\}  ^{m}\right)  ^{N}$\ such that $Y\in A_{x,\delta}$\ with
probability at least $1-\delta$. \ Then$\ S\in\mathsf{SampP}^{B}$.
\end{enumerate}

\begin{proof}
[Proof of Theorem \ref{samprel}\ (Sampling/Searching Equivalence Theorem)]For
part (i), given an input $\left\langle x,0^{1/\delta}\right\rangle $, suppose
we want to output an $N$-tuple $Y=\left\langle y_{1},\ldots,y_{N}\right\rangle
\in\left(  \left\{  0,1\right\}  ^{m}\right)  ^{N}$\ such that $Y\in
A_{x,\delta}$, with success probability at least $1-\delta$. \ Recall that
$N=m/\delta^{2.1}$. \ Then the algorithm is this:

\begin{enumerate}
\item[(1)] Set $\varepsilon:=\frac{\delta}{2N}=\frac{\delta^{3.1}}{2m}$.

\item[(2)] Call $\mathcal{O}$ on inputs $\left\langle x,0^{1/\varepsilon
},r_{1}\right\rangle ,\ldots,\left\langle x,0^{1/\varepsilon},r_{N}%
\right\rangle $, where $r_{1},\ldots,r_{N}$\ are independent random strings,
and output the result as $Y=\left\langle y_{1},\ldots,y_{N}\right\rangle $.
\end{enumerate}

Clearly this algorithm runs in $\operatorname*{poly}\left(  n,1/\delta\right)
$\ time. \ Furthermore, by Lemma \ref{stor}, its failure probability is at
most%
\[
\varepsilon N+\frac{b}{2^{\beta}}\leq\delta.
\]

For part (ii), given an input $\left\langle x,0^{1/\varepsilon}\right\rangle
$, suppose we want to sample from a distribution $\mathcal{C}_{x}$\ such that
$\left\Vert \mathcal{C}_{x}-\mathcal{D}_{x}\right\Vert \leq\varepsilon$.
\ Then the algorithm is this:

\begin{enumerate}
\item[(1)] Set $\delta:=\varepsilon/2$, so that $N=m/\delta^{2.1}%
=\Theta\left(  m/\varepsilon^{2.1}\right)  $.

\item[(2)] Call $B$ on input $\left\langle x,0^{1/\delta}\right\rangle $, and
let $Y=\left\langle y_{1},\ldots,y_{N}\right\rangle $\ be $B$'s output.

\item[(3)] Choose $i\in\left[  N\right]  $\ uniformly at random, and output
$y_{i}$\ as the sample from $\mathcal{C}_{x}$.
\end{enumerate}

Clearly this algorithm runs in $\operatorname*{poly}\left(  n,1/\varepsilon
\right)  $\ time. \ Furthermore, by Lemma \ref{rtos} we have%
\begin{align*}
\left\Vert \mathcal{C}_{x}-\mathcal{D}_{x}\right\Vert  &  \leq\delta
+Q_{B}\sqrt{\frac{\beta}{N}}\\
&  \leq\frac{\varepsilon}{2}+Q_{B}\sqrt{\frac{\varepsilon^{2.1}\left(
2+\log1/\varepsilon\right)  }{m}},
\end{align*}
for some constant $Q_{B}$\ depending only on $B$. \ So in particular, there
exists a constant $C_{B}$\ such that $\left\Vert \mathcal{C}_{x}%
-\mathcal{D}_{x}\right\Vert \leq\varepsilon$\ for all $m\geq C_{B}$. \ For
$m<C_{B}$, we can simply hardwire a description of $\mathcal{D}_{x}$\ for
every $x$\ into the algorithm (note that the algorithm can depend on $B$; we
do not need a single algorithm that works for all $B$'s simultaneously).
\end{proof}

In particular, Theorem \ref{samprel} means that $S\in\mathsf{SampP}$\ if and
only if $R_{S}\in\mathsf{FBPP}$, and likewise $S\in\mathsf{SampBQP}$\ if and
only if $R_{S}\in\mathsf{FBQP}$, and so on for any model of computation that
is \textquotedblleft below recursive\textquotedblright\ (i.e., simulable by a
Turing machine) and has the extremely simple closure properties used in the proof.

\subsection{Implication for Quantum Computing\label{IMPQC}}

We now apply Theorem \ref{samprel}\ to\ prove Theorem \ref{fbqpthm}, that
$\mathsf{SampP}=\mathsf{SampBQP}$ if and only if $\mathsf{FBPP}=\mathsf{FBQP}$.

\begin{proof}
[Proof of Theorem \ref{fbqpthm}]First, suppose $\mathsf{SampP}%
=\mathsf{SampBQP}$. \ Then consider a search problem $R=\left(  A_{x}\right)
_{x}$ in $\mathsf{FBQP}$. \ By assumption, there exists a polynomial-time
quantum algorithm $Q$ that, given $\left\langle x,0^{1/\delta}\right\rangle
$\ as input, outputs a $y$\ such that $y\in A_{x}$ with probability at least
$1-\delta$. \ Let $\mathcal{D}_{x,\delta}$\ be the probability distribution
over $y$'s output by $Q$ on input $\left\langle x,0^{1/\delta}\right\rangle $.
\ Then to solve $R$ in $\mathsf{FBPP}$, clearly it suffices to sample
approximately from $\mathcal{D}_{x,\delta}$ in classical polynomial time.
\ But we can do this by the assumption that $\mathsf{SampP}=\mathsf{SampBQP}%
$.\footnote{As mentioned in Section \ref{INTRO}, the same argument shows that
$\mathsf{SampP}=\mathsf{SampBQP}$\ (or equivalently, $\mathsf{FBPP}%
=\mathsf{FBQP}$) implies $\mathsf{BPP}=\mathsf{BQP}$. \ However, the converse
is far from clear: we have no idea whether $\mathsf{BPP}=\mathsf{BQP}%
$\ implies $\mathsf{SampP}=\mathsf{SampBQP}$.}

Second, suppose $\mathsf{FBPP}=\mathsf{FBQP}$. \ Then consider a sampling
problem $S$\ in $\mathsf{SampBQP}$. \ By Theorem \ref{samprel}, we can define
a search counterpart $R_{S}$ of $S$, such that%
\begin{align*}
S\in\mathsf{SampBQP}  &  \Longrightarrow R_{S}\in\mathsf{FBQP}\\
&  \Longrightarrow R_{S}\in\mathsf{FBPP}\\
&  \Longrightarrow S\in\mathsf{SampP}.
\end{align*}
Hence $\mathsf{SampP}=\mathsf{SampBQP}$.
\end{proof}

Theorem \ref{fbqpthm} is easily seen to relativize: for all oracles $A$, we
have $\mathsf{SampP}^{A}=\mathsf{SampBQP}^{A}$\ if and only if $\mathsf{FBPP}%
^{A}=\mathsf{FBQP}^{A}$. \ (Of course, when proving a relativized version of
Theorem \ref{samprel}, we have to be careful to define the search problem
$R_{S}$\ using Kolmogorov complexity for Turing machines with $A$-oracles.)

\section{Extensions and Open Problems\label{OPEN}}

\subsection{Equivalence of Sampling and \textit{Decision}
Problems?\label{SAMPDEC}}

Perhaps the most interesting question we leave open is whether any nontrivial
equivalence holds between sampling (or search) problems on the one hand, and
\textit{decision} or \textit{promise} problems on the other. \ In Theorem
\ref{samprel}, it was certainly essential to consider large numbers of
outputs; we would have no idea how to prove an analogous result with a promise
problem $P_{S}$\ or language $L_{S}$\ instead of the search problem $R_{S}$.

One way to approach this question is as follows: does there exist a sampling
problem $S$ that is provably \textit{not} equivalent to any decision problem,
in the sense that for every language $L\subseteq\left\{  0,1\right\}  ^{\ast}%
$, either $S\notin\mathsf{SampP}^{L}$, or else there exists an oracle
$\mathcal{O}$\ solving $S$\ such that $L\notin\mathsf{BPP}^{\mathcal{O}}$?
\ What if we require the oracle $\mathcal{O}$\ to be computable? \ As far as
we know, these questions are open.

One might object that,\ given any sampling problem $S$, it is easy to define a
language $L_{S}$\ that is \textquotedblleft equivalent\textquotedblright\ to
$S$, by using the following simple enumeration trick. \ Let $M_{1}%
,M_{2},\ldots$\ be an enumeration of probabilistic Turing machines with
polynomial-time alarm clocks. \ Given a sampling problem $S=\left(
\mathcal{D}_{x}\right)  _{x\in\left\{  0,1\right\}  ^{\ast}}$ and an input
$X=\left\langle x,0^{1/\varepsilon}\right\rangle $, say that $M_{t}%
$\ \textit{succeeds} on $X$ if $M_{t}\left(  X\right)  $ samples from a
distribution $\mathcal{C}_{X}$\ such that $\left\Vert \mathcal{C}%
_{X}-\mathcal{D}_{x}\right\Vert \leq\varepsilon$. \ Also, if $x$ is an $n$-bit
string, define the \textit{length} of $X=\left\langle x,0^{1/\varepsilon
}\right\rangle $\ to be $\ell\left(  X\right)  :=n+1/\varepsilon$.

We now define a language $L_{S}\subseteq\left\{  0,1\right\}  ^{\ast}$. \ For
all $n$, let $M_{t\left(  n\right)  }$\ be the lexicographically first $M_{t}$
that succeeds on\ all inputs $X$\ such that $\ell\left(  X\right)  \leq n$.
\ Then for all $y\in\left\{  0,1\right\}  ^{n}$, we set $y\in L_{S}$\ if and
only if the Turing machine encoded by $y$ halts in at most $n^{t\left(
n\right)  }$\ steps when run on a blank tape.

\begin{proposition}
\label{stupid}$S\in\mathsf{SampP}$ if and only if $L_{S}\in\mathsf{P}$.
\end{proposition}

\begin{proof}
First suppose $S\in\mathsf{SampP}$. \ Then there exists a polynomial-time
Turing machine that succeeds on every input $X=\left\langle x,0^{1/\varepsilon
}\right\rangle $. \ Let $M_{t}$\ be the lexicographically first such machine.
\ Then it is not hard to see that $L_{S}$\ consists of a finite prefix,
followed by the $n^{t}$-time bounded halting problem. \ Hence $L_{S}%
\in\mathsf{P}$.

Next suppose $S\notin\mathsf{SampP}$. \ Then \textit{no} machine $M_{t}%
$\ succeeds on every input $X$,\ so $t\left(  n\right)  $\ grows without bound
as a function of $n$. \ By standard diagonalization arguments, the
$n^{t\left(  n\right)  }$-time\ bounded halting problem is not in $\mathsf{P}%
$\ for any $t$\ that grows without bound, regardless of whether $t$ is
time-constructible. \ Therefore $L_{S}\notin\mathsf{P}$.
\end{proof}

Admittedly, Proposition \ref{stupid}\ feels like cheating---but \textit{why}
exactly is it cheating? \ Notice that we \textit{did} give a procedure to
decide whether $y\in L_{S}$ for any input $y$. \ This fact makes Proposition
\ref{stupid}\ at least \textit{somewhat} more interesting than the
\textquotedblleft tautological\textquotedblright\ way to ensure $S\in
\mathsf{SampP}\Longleftrightarrow L_{S}\in\mathsf{P}$:

\begin{quotation}
\noindent\textit{\textquotedblleft Take }$L_{S}$\textit{\ to be the empty
language if} $S\in\mathsf{SampP}$\textit{, or an }$\mathsf{EXP}$%
\textit{-complete language if }$S\notin\mathsf{SampP}$!\textquotedblright
\end{quotation}

\noindent In our view, the real problem with Proposition \ref{stupid} is that
it uses enumeration of Turing machines to avoid the need to \textit{reduce}
the sampling problem$\ S$\ to the language $L_{S}$\ or vice versa.\ \ Of
course, Theorem \ref{samprel} did not quite reduce $S$ to the search problem
$R_{S}$ either. \ However, Theorem \ref{samprel} came \textquotedblleft close
enough\textquotedblright\ to giving a reduction that we were able to use it to
derive interesting consequences for complexity theory, such as $\mathsf{SampP}%
=\mathsf{SampBQP}$ if and only if $\mathsf{FBPP}=\mathsf{FBQP}$. \ If we
attempted to prove similar consequences from Proposition \ref{stupid}, then we
would end up with a \textit{different} language $L_{S}$, depending on whether
our starting assumption was $S\in\mathsf{SampP}$, $S\in\mathsf{SampBQP}$, or
some other assumption. \ By contrast, Theorem \ref{samprel}\ constructed a
\textit{single} search problem $R_{S}$\ that is equivalent to $S$ in the
classical model, the quantum model, and every other \textquotedblleft
reasonable\textquotedblright\ computational model.

\subsection{Was Kolmogorov Complexity Necessary?\label{NECKOLMOG}}

Could we have proved Theorem \ref{samprel} \textit{without} using Kolmogorov
complexity or anything like it, and without making a computability assumption
on the oracle for $R_{S}$? \ One way to formalize this question is to ask the
analogue of our question from Section \ref{SAMPDEC}, but this time for
sampling versus \textit{search} problems. \ In other words, does there exist a
sampling problem $S$ such that, for every search problem $R$, either there
exists an oracle $\mathcal{O}$\ solving $S$\ such that $R\notin\mathsf{FBPP}%
^{\mathcal{O}}$, or there exists an oracle $\mathcal{O}$\ solving $R$\ such
that $S\notin\mathsf{SampP}^{\mathcal{O}}$? \ Notice that, if $R$ is the
search problem from Theorem \ref{samprel}, then the latter oracle (if it
exists) must be uncomputable. \ Thus, we are essentially asking whether the
computability assumption in Theorem \ref{samprel}\ was necessary.

\subsection{From Search Problems to Sampling Problems\label{RELTOSAMP}}

Theorem \ref{samprel} showed how to take any sampling problem $S$, and define
a search problem $R=R_{S}$\ that is equivalent to $S$. \ Can one go the other
direction? \ That is, given a search problem $R$, can one define a sampling
problem $S=S_{R}$\ that is equivalent to $R$? \ The following theorem is the
best we were able to find in this direction.

\begin{theorem}
\label{otherdir}Let $R=\left(  A_{x}\right)  _{x}$\ be any search problem.
\ Then there exists a sampling problem $S_{R}=\left\{  \mathcal{D}%
_{x}\right\}  _{x}$\ that is \textquotedblleft almost
equivalent\textquotedblright\ to $R$, in the following senses.

\begin{enumerate}
\item[(i)] If $\mathcal{O}$\ is any oracle solving $S_{R}$, then
$R\in\mathsf{FBPP}^{\mathcal{O}}$.

\item[(ii)] If $B$ is any probabilistic Turing machine solving $R$, then there
exists a constant $\eta_{B}>0$ such that a $\mathsf{SampP}^{B}$ machine can
sample from a probability distribution $\mathcal{C}_{x}$\ with $\left\Vert
\mathcal{C}_{x}-\mathcal{D}_{x}\right\Vert \leq1-\eta_{B}$.
\end{enumerate}
\end{theorem}

\begin{proof}
Let $\mathcal{U}_{x}$\ be the universal prior, in which every string
$y$\ occurs with probability at least $c\cdot2^{-K\left(  y|x\right)  }$, for
some constant $c>0$. \ Then to define the sampling problem $S_{R}$, we let
$\mathcal{D}_{x}$\ be the distribution obtained by drawing $y\sim
\mathcal{U}_{x}$\ and then conditioning on the event $y\in A_{x}$. \ (Note
that $\mathcal{D}_{x}$ is well-defined, since $\mathcal{U}_{x}$\ assigns
nonzero probability to every $y$.)

For (i), notice that $\mathcal{D}_{x}$\ has support only on $A_{x}$. \ So if
we can sample a distribution $\mathcal{C}_{x}$\ such that $\left\Vert
\mathcal{C}_{x}-\mathcal{D}_{x}\right\Vert \leq\varepsilon$, then certainly we
can output an element of $A_{x}$\ with probability at least $1-\varepsilon$.

For (ii), let $\mathcal{C}_{x,\delta}$\ be the distribution over values of
$B\left(  x,0^{1/\delta},r\right)  $\ induced by varying the random string
$r$. \ Then we claim that $\left\Vert \mathcal{C}_{x,\delta}-\mathcal{D}%
_{x}\right\Vert \leq1-\Omega\left(  1\right)  $, so long as $\delta\leq
\Delta_{B}$ for some constant $\Delta_{B}$\ depending on $B$. \ To see this,
first let $\mathcal{C}^{\prime}$\ be the distribution obtained by drawing
$y\sim\mathcal{C}_{x,\delta}$\ and then conditioning on the event $y\in A_{x}%
$. \ Then since $\Pr_{y\sim\mathcal{C}_{x,\delta}}\left[  y\in A_{x}\right]
\geq1-\delta$, we have $\left\Vert \mathcal{C}^{\prime}-\mathcal{C}_{x,\delta
}\right\Vert \leq\delta$.

Now let $q_{y}:=\Pr_{\mathcal{C}^{\prime}}\left[  y\right]  $. \ Then by Lemma
\ref{kolmogsmall}, there exists a constant $g_{B}$\ depending on $B$ such that%
\[
q_{y}\leq g_{B}\cdot2^{-K\left(  y|x\right)  }%
\]
for all $y\in A_{x}$. \ On the other hand, let $p_{y}:=\Pr_{\mathcal{D}_{x}%
}\left[  y\right]  $\ and $u_{y}:=\Pr_{\mathcal{U}_{x}}\left[  y\right]  $.
\ Then there exists a constant $\alpha\geq1$\ such that $p_{y}=\alpha u_{y}%
$\ if $y\in A_{x}$\ and $p_{y}=0$\ otherwise. \ So%
\[
p_{y}\geq u_{y}\geq c\cdot2^{-K\left(  y|x\right)  }%
\]
for all $y\in A_{x}$. \ Hence $p_{y}\geq\frac{c}{g_{B}}q_{y}$, so%
\begin{align*}
\left\Vert \mathcal{C}^{\prime}-\mathcal{D}_{x}\right\Vert  &  =\sum_{y\in
A_{x}~:~p_{y}<q_{y}}\left\vert p_{y}-q_{y}\right\vert \\
&  \leq1-\frac{c}{g_{B}}.
\end{align*}
Therefore%
\begin{align*}
\left\Vert \mathcal{C}_{x,\delta}-\mathcal{D}_{x}\right\Vert  &
\leq\left\Vert \mathcal{C}_{x,\delta}-\mathcal{C}^{\prime}\right\Vert
+\left\Vert \mathcal{C}^{\prime}-\mathcal{D}_{x}\right\Vert \\
&  \leq1-\frac{c}{g_{B}}+\delta,
\end{align*}
which is $1-\Omega_{B}\left(  1\right)  $\ provided $\delta\leq\frac{c}%
{2g_{B}}$.
\end{proof}

We see it as an interesting problem whether Theorem \ref{otherdir}\ still
holds with the condition $\left\Vert \mathcal{C}_{x}-\mathcal{D}%
_{x}\right\Vert \leq1-\eta_{B}$\ replaced by $\left\Vert \mathcal{C}%
_{x}-\mathcal{D}_{x}\right\Vert \leq\varepsilon$\ (in other words, with
$S_{R}\in\mathsf{SampP}^{B}$).

\subsection{Making the Search Problem\ Checkable\label{CHECK}}

One obvious disadvantage of Theorem \ref{samprel} is that the search problem
$R=\left(  A_{x}\right)  _{x}$ is defined using Kolmogorov complexity, which
is uncomputable. \ In particular, there is no algorithm to decide whether
$y\in A_{x}$. \ However, it is not hard to fix this problem, by replacing the
Kolmogorov complexity with the \textit{time-bounded} or \textit{space-bounded}
Kolmogorov complexities in our definition of $R$. \ The price is that we then
also have to assume a complexity bound on the Turing machine $B$\ in the
statement of Theorem \ref{samprel}. \ In more detail:

\begin{theorem}
\label{samprel2}Let $S$\ be any sampling problem, and let $f$ be a
time-constructible function. \ Then there exists a search problem
$R_{S}=\left(  A_{x}\right)  _{x}$\ such that

\begin{enumerate}
\item[(i)] If $\mathcal{O}$\ is any oracle solving $S$, then $R_{S}%
\in\mathsf{FBPP}^{\mathcal{O}}$.

\item[(ii)] If $B$ is any $\mathsf{BPTIME}\left(  f\left(  n\right)  \right)
$ Turing machine solving $R_{S}$, then $S\in\mathsf{SampP}^{B}$.

\item[(iii)] There exists a $\mathsf{SPACE}\left(  f\left(  n\right)
+n^{O\left(  1\right)  }\right)  $\ algorithm to decide whether $y\in A_{x}$,
given $x$\ and $y$.
\end{enumerate}
\end{theorem}

\begin{proof}
[Proof Sketch]The proof is almost the same as the proof of Theorem
\ref{samprel}. \ Let $T:=f\left(  n\right)  +n^{O\left(  1\right)  }$, and
given a string $y$, let $K_{\operatorname*{SPACE}\left(  T\right)  }\left(
y\right)  $\ be the $T$-space bounded\ Kolmogorov complexity of $y$. \ Then
the only real difference is that, when defining the search problem $R_{S}$, we
replace the conditional Kolmogorov complexity $K\left(  Y~|~x,\delta\right)  $
by the space-bounded complexity $K_{\operatorname*{SPACE}\left(  T\right)
}\left(  Y~|~x,\delta\right)  $. \ This ensures that property (iii) holds.

Certainly property (i) still holds, since it only used the fact that there are
few tuples $Y\in\left(  \left\{  0,1\right\}  ^{m}\right)  ^{N}$\ with small
Kolmogorov complexity, and that is still true for space-bounded Kolmogorov complexity.

For property (ii), it suffices to observe that Lemma \ref{kolmogsmall} has the
following \textquotedblleft effective\textquotedblright\ version. \ Let
$\mathcal{D}=\left\{  p_{y}\right\}  $\ be any distribution over strings that
is samplable in $\mathsf{BPTIME}\left(  f\left(  n\right)  \right)  $, and let
$y$\ be any element in the support of $\mathcal{D}$. \ Then there exists a
constant $C_{\mathcal{D}}$,\ depending on $\mathcal{D}$,\ such that%
\[
K_{\operatorname*{SPACE}\left(  T\right)  }\left(  y\right)  \leq\log_{2}%
\frac{1}{p_{y}}+C_{\mathcal{D}}.
\]
The proof is simply to notice that, in $\mathsf{SPACE}\left(  f\left(
n\right)  +n^{O\left(  1\right)  }\right)  $, we can compute the probability
$p_{y}$ of \textit{every} $y$ in the support of $\mathcal{D}$, and can
therefore recover any particular string $y$ from its Shannon-Fano code. \ This
means that the analogue of Lemma \ref{rtos}\ goes through, as long as $B$\ is
a $\mathsf{BPTIME}\left(  f\left(  n\right)  \right)  $\ machine.
\end{proof}

In Theorem \ref{samprel2}, how far can we decrease the computational
complexity of $R_{S}$? \ It is not hard to replace the upper bound of
$\mathsf{SPACE}\left(  f\left(  n\right)  +n^{O\left(  1\right)  }\right)
$\ by $\mathsf{CH}\left(  f\left(  n\right)  +n^{O\left(  1\right)  }\right)
$\ (where $\mathsf{CH}$\ denotes the counting hierarchy), but can we go
further? \ It seems unlikely that one could check in $\mathsf{NP}$\ (or
$\mathsf{NTIME}\left(  f\left(  n\right)  +n^{O\left(  1\right)  }\right)  $)
whether $y\in A_{x}$, for a search problem $R_{S}=\left\{  A_{x}\right\}
_{x}$\ equivalent to $S$, but can we give formal evidence against this possibility?

\section{Acknowledgments}

I thank Alex Arkhipov for helpful discussions that motivated this work, and
Dana Moshkovitz for pointing me to Proposition \ref{anup}\ from \cite{rao}.

\bibliographystyle{plain}
\bibliography{thesis}

\end{document}